\documentclass[sigconf]{aamas}  

\usepackage{booktabs}

\usepackage{flushend}
\acmDOI{doi}  
\acmISBN{}  
\acmConference[AAMAS'20]{Proc.\@ of the 19th International Conference on Autonomous Agents and Multiagent Systems (AAMAS 2020), B.~An, N.~Yorke-Smith, A.~El~Fallah~Seghrouchni, G.~Sukthankar (eds.)}{May 2020}{Auckland, New Zealand}  
\acmYear{2020}  
\copyrightyear{2020}  
\acmPrice{}  

\usepackage{amsmath}
\usepackage{amssymb}
\usepackage{enumitem}
\usepackage{appendix}
\usepackage{amsthm}
\usepackage{cleveref}
\usepackage[ruled]{algorithm2e}
\usepackage{subfig}
\usepackage{paralist}
\usepackage{bbm}
\usepackage{tikz}
\usetikzlibrary{backgrounds}

\usepackage{graphicx}
\usepackage{amsmath}
\usepackage{tabularx}
\usepackage{amsthm}
\usepackage{booktabs}
\usepackage{algorithmic}
\usepackage{amsfonts}
\usepackage{comment}
\usepackage{amssymb}
\usepackage{appendix}
\usepackage{amsmath}
\usepackage{graphicx} 
\usepackage{wrapfig}
\usepackage{amsthm}
\usepackage{thmtools}

\newcommand{\OnCTRSD}{{\sc On-CT-RSD}\xspace}
\newcommand{\OnCARSD}{{\sc On-CA-RSD}\xspace}
\newcommand{\CABPRSD}{{\sc CA-BP-RSD}\xspace}
\newcommand{\CABQRSD}{{\sc CA-BQ-RSD}\xspace}
\newcommand{\CABERSD}{{\sc CA-BE-RSD}\xspace}
\newcommand{\FFCTRSD}{{\sc FF-CT-RSD}\xspace}

\newcommand{\OnCARSDstar}{{\sc On-CA-RSD$^*$}\xspace}
\newcommand{\FFCTRSDstar}{{\sc FF-CT-RSD$^*$}\xspace}

\graphicspath{ {img/} }

\newcommand{\HAX}{HA+X\xspace}

\newcommand{\V}{\mathcal{V}}
\newcommand{\N}{\mathcal{N}}
\renewcommand{\AA}{\mathcal{A}}
\newcommand{\A}{A}
\renewcommand{\G}{\mathcal{G}}
\newcommand{\E}{\mathcal{E}}

\newcommand{\Q}{\mathbb{Q}}
\newcommand{\SW}{\text{SW}}
\newcommand{\OPT}{\text{OPT}}
\newcommand{\Exp}{\mathbb{E}}

\newcommand{\tup}[1]{\langle #1 \rangle}
\newcommand{\rsd}{\textsc{RSD}\xspace}
\newcommand{\rsdstar}{\textsc{RSD}$^*$\xspace}
\newcommand{\sd}{\textsc{SD}\xspace}

\mathchardef\mhyphen="2D 

\renewcommand{\phi}{\varphi} 

\begin{document}

\title{Keeping Your Friends Close: Land Allocation with Friends}  



%
\author{Edith Elkind}
\affiliation{%
 \institution{University of Oxford}
 \streetaddress{Wolfson Building, Parks Road
}
 \city{Oxford} 
 \state{United Kingdom} 
 \postcode{OX1 3QD}
}
\email{eelkind@gmail.com}

\author{Neel Patel}
\affiliation{%
 \institution{National University of Singapore}
 \streetaddress{13 Computing Drive 
}
 \city{} 
 \state{Singapore} 
 \postcode{117417}
}
\email{neeltuk@gmail.com}

\author{Alan Tsang}
\affiliation{%
 \institution{National University of Singapore}
 \streetaddress{13 Computing Drive 
}
 \city{} 
 \state{Singapore} 
 \postcode{117417}
}
\email{akhtsang@gmail.com}

\author{Yair Zick}
\affiliation{%
 \institution{National University of Singapore}
 \streetaddress{13 Computing Drive 
}
 \city{} 
 \state{Singapore} 
 \postcode{117417}
}
\email{dcsyaz@nus.edu.sg}

%
%
%
%
%
%
%
\begin{abstract}
    We examine the problem of assigning plots of land to prospective buyers who prefer living next to their friends. They care not only about the plot they receive, but also about their neighbors. This externality results in a highly non-trivial problem structure, as both friendship and land value play a role in determining agent behavior. 
    We examine mechanisms that guarantee truthful reporting of \emph{both} land values and friendships. We propose variants of {\em random serial dictatorship} (\rsd) that can offer both truthfulness and welfare guarantees. Interestingly, our social welfare guarantees are parameterized by the value of friendship: if these values are low, enforcing truthful behavior results in poor welfare guarantees and imposes significant constraints on agents' choices; if they are high, we achieve 
    good approximation to the optimal social welfare. 
\end{abstract}

\maketitle

\section{Introduction}\label{sec:intro}
A village in a quaint part of country X recently received a permission to expand. Predetermined plots of land, of approximately equal size and price, have been drawn and must be assigned to prospective buyers\footnote{As it happens, one of these buyers is the sister of an author.}. However, while similar in size and official value, plots are not viewed as identical by the buyers: some buyers prefer living close to the village center, others favor living in an area with a view of the surrounding mountains, and yet others are interested in level plots amenable to a home garden. Land ownership laws preclude direct ownership by buyers; rather, land is leased from a central governing body, and prospective buyers are prohibited by law from paying each other in order to secure land plots. In other words, land plots are to be treated as {\em indivisible goods}, and are to be allocated without monetary transfers. Prospective buyers form a small, close-knit community. Several of them are siblings (with parents having lived in the village for decades), or are long-term residents (in rented properties), with friends they'd like to be close to. Consequently, buyers have preferences not just over plots, but also over their potential neighbors. In fact, some pairs of buyers only care about being neighbors, regardless of where they end up. Thus, we are interested in  {\em mechanisms that would enable the buyers to distribute the plots among themselves in a fair and efficient manner, and account for friendships}.

\begin{figure}[ht]

  \centering
  \includegraphics[width=.8\linewidth]{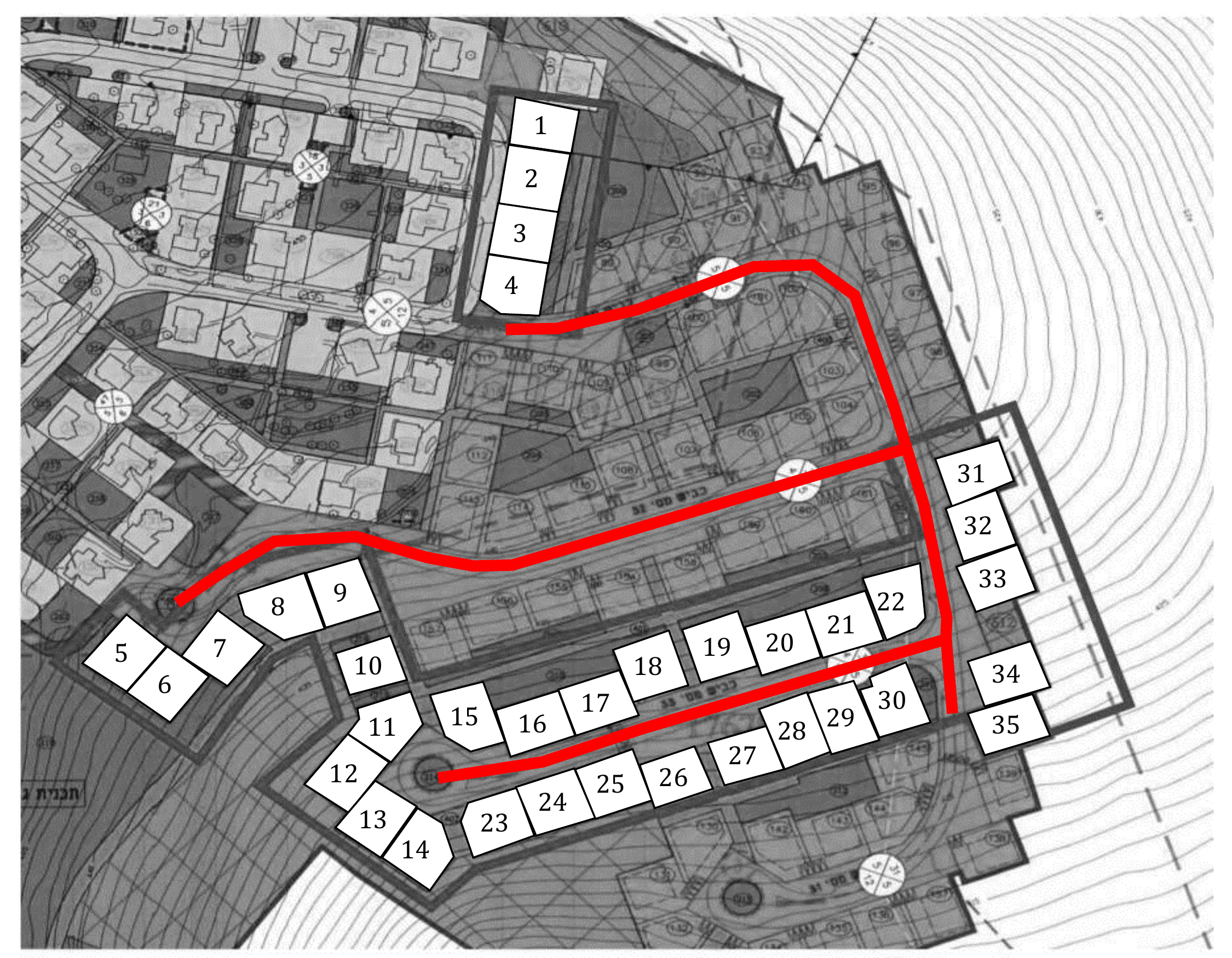}  
  \caption{Map of the proposed village expansion (plots are numbered $1-35$). Red lines denote roads.}
  \label{fig:mapofvillage}
\end{figure}
\begin{figure}[ht]
  \centering
  \includegraphics[width=\linewidth]{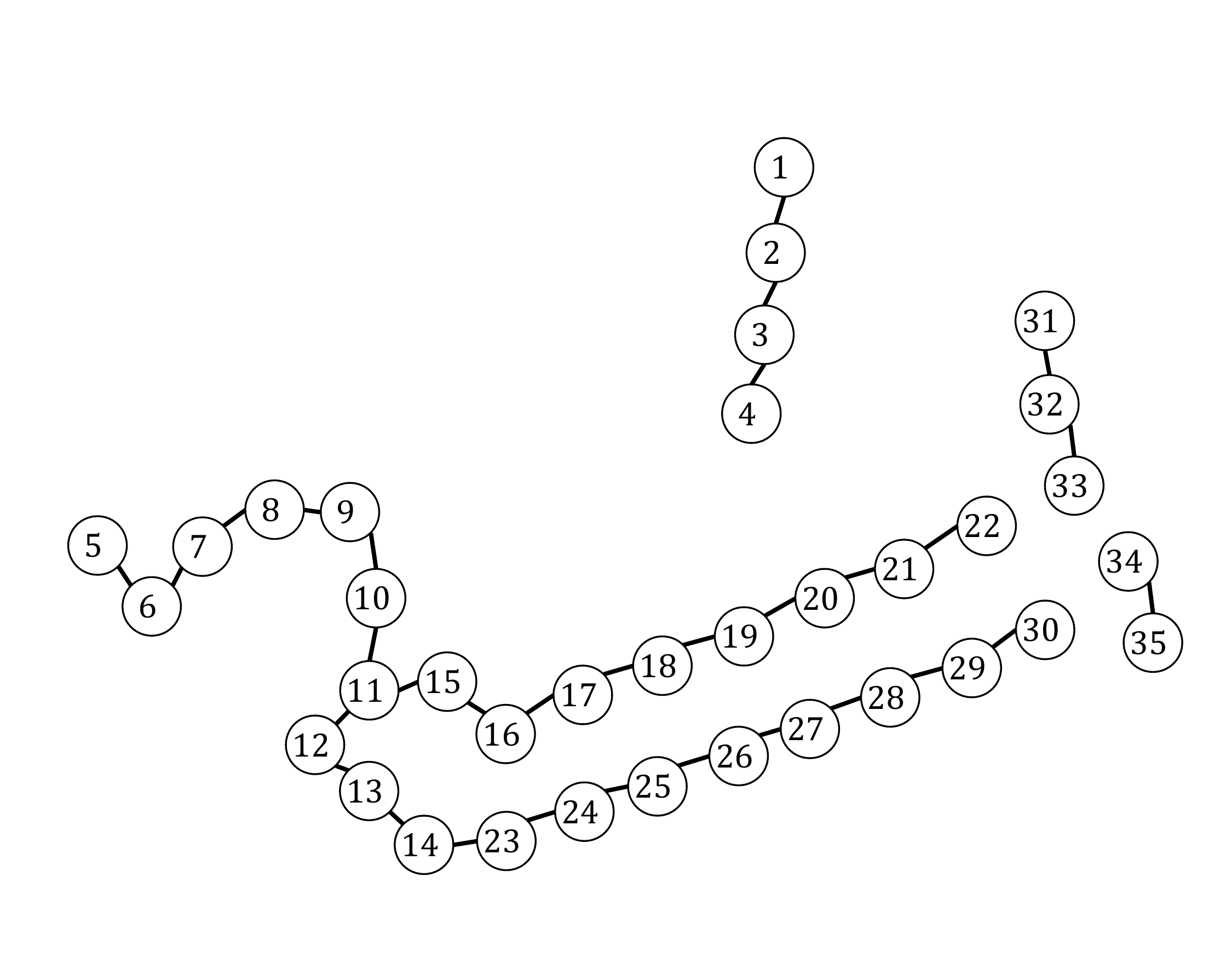}  

\caption{The plot graph based on Figure \ref{fig:mapofvillage}. Two plots are adjacent if they share a border}
\label{fig:plotgraph}
\end{figure}

\subsection{Our Contributions}\label{sec:contrib}
We briefly discuss the complexity of finding an allocation that maximizes the social welfare in the complete information scenario, showing that this problem is NP-hard. We then focus on the setting where each agent has at most one friend. This constraint is both realistic and simplifies our computational problem significantly: while our problem remains NP-hard, it admits a $2$-approximation algorithm in this case. 

We then investigate our problem from the perspective of mechanism design without money: can we incentivize agents to truthfully report their plot values and friendship information? Given our application domain, we are interested in mechanisms that are simple to describe and participate in, while providing good social welfare guarantees. Since our problem generalizes the one-sided matching problem \cite{HZ79}, a natural starting point is the Random Serial Dictatorship (RSD) mechanism, under which agents pick plots one by one. We establish that \rsd does not perform well in the presence of friendships, and explore several modifications of \rsd, in which the picking order is based on friendship information. We identify settings in which our mechanisms are truthful and produce Pareto optimal outcomes, and provide bounds on their expected social welfare for the case where agents have binary valuations for the plots. 

\subsection{Related Work}\label{sec:related}
%
One-sided matching markets have been studied for several decades. \citet{HZ79} propose a Pareto optimal, envy-free mechanism, which is, however, not truthful. \citet{svensson1999strategy} shows that the Random Serial Dictatorship (RSD) is the only truthful mechanism that satisfies ex-post Pareto optimality, anonymity and non-bossiness.

The social welfare of truthful mechanisms in one-sided matching markets has been studied by \citet{bhalgat2011social} for rank-based valuation functions. \citet{filos2014social} consider the social welfare of RSD for unit sum preferences, and show that RSD offers a $\sqrt{n}$-approximation to the optimal social welfare in this case. \citet{adamczyk2014efficiency} focus on binary and unit-range preferences, and show that RSD offers a $3$-approximation to the optimal social welfare for binary preferences, and a $\sqrt{en}$-approximation for unit range preferences. \citet{christodoulou2016social} analyze the Price of Anarchy (PoA) of one-sided matching markets for unit-sum preferences. They show that PoA for RSD is $\mathcal{O}(\sqrt n)$. \citet{krysta2016house} study the one-sided matching market problem under matroid constraints. They propose a truthful mechanism and show that it offers $\frac{\mathrm e}{\mathrm e - 1}$-approximation of the optimal social welfare.

\citet{bodine2011peer} analyze a housing allocation problem where students have an inherent friendship structure. 
They focus on allocation stability and social welfare, rather than strategic behavior. An online variant of this problem is studied by \citet{bei17}. 
\citet{Massand2019graphical} also consider the stability of a one-sided matching market with externalities, but assume that agents cannot misreport their valuations. 

\section{Model and Preliminaries}
We omit several technical proofs from the paper due to page limits; these will appear in a full version of this work.


We consider a set of agents $N = \{1, \ldots n\}$ (land buyers) who need to be matched to $n$ plots $\V = \{v_1, \ldots v_n\}$.  
Each agent receives exactly one plot. Thus, the goal is to output an {\em allocation}, i.e., a bijection $\A : N \rightarrow \V$, where agent $i$ gets plot $\A(i)$.

We represent neighboring plots using a {\em plot graph} $\G =\tup{\V,\E}$: this is an undirected graph where \emph{nearby} plots $w$ and $v$ are connected by an edge $\{w,v\} \in \E$ (see Figure \ref{fig:mapofvillage}).
Each agent $i\in N$ has a {\em valuation function} $u_i:\V\to[0,1]\cap{\mathbb Q}$: $u_i(v)$ is the value $i$ derives from receiving plot $v$. 
Agents have friends, and care about living next to them. We represent friendships by a weighted directed {\em friendship graph} $\tup{N, F}$, where
$(i, j)\in F$ indicates that $i$ and $j$ are friends and the edge weight $\phi_{i, j}\in\Q_{\ge 0}$ is the additional utility $i$ obtains for living next to $j$. 
We assume that friendships are reciprocal, but not necessarily symmetric; i.e., $(i, j)\in F \Leftrightarrow (j, i)\in F$, but it may be the case that $\phi_{i, j}\neq \phi_{j, i}$.
Let $F^*=\left\{\{i, j\}: (i, j)\in F\right\}$; the unweighted undirected graph $\tup{N, F^*}$ captures the presence of friendships, but not their weights. We set $\phi_{\min}=\min_{(i, j)\in F}\phi_{i, j}$. The quantity $\phi_{\min}$ plays an important role in our analysis: some of our proposed mechanisms offer better performance guarantees when $\phi_{\min}>1$, i.e. when the value of friendship exceeds the value of any plot.



The {\em utility} $U_i(\A)$ of agent $i$ under allocation $\A$ is  
\begin{equation}
u_i(\A(i)) + \sum_{(i, j)\in F} \phi_{i,j}\times \mathbb{I}\left(\{\A(i),\A(j)\} \in \E\right).\label{eq:utility}
\end{equation}  
The first term in \eqref{eq:utility} is agent $i$'s utility from the plot she receives; the second term is her (non-negative) \emph{externality} for nearby friends. 

An instance of our allocation problem is a tuple 
$$
I=\tup{N, \V, \E, F, (u_i)_{i\in N}, (\phi_{i, j})_{(i, j)\in F)}};
$$
let $\AA(I)$ denote the set of all allocations for an instance $I$.

The {\em social welfare} of an allocation $\A\in \AA(I)$ is defined as the sum of agents' utilities: $\SW(\A)=\sum_{i\in N}U_i(\A)$.
Let $\OPT(I)=\max_{\A\in\AA(I)}\SW(\A)$.
Given two allocations $\A, \A'\in \AA(I)$, we say that $\A'$ {\em dominates} $\A$ if $U_i(\A') \ge U_i(\A)$ for all $i\in N$, and the inequality is strict for at least one agent.
An allocation $\A$ is {\em Pareto optimal (PO)} if no other allocation dominates it. A non-PO allocation presents an avoidable loss of social welfare; we are thus interested in algorithms that output PO allocations.

We consider several special cases of our problem. 
We say that an instance $I$ is {\em friendship-uniform} if there exists a positive value $\phi\in\Q_{\ge 0}$ such that $\phi_{i, j}=\phi$ for all $(i, j)\in F$. We say that $I$ is {\em binary}
if $u_i(v)\in \{0, 1\}$ for all $i\in N$, $v\in V$. We say that $I$ is {\em generic} if for every $i\in N$, every pair of plots $v, w\in\V$ and every edge $(i, j)\in F$ we have $u_i(v)\neq u_i(w)$, $u_i(v)\neq u_i(w)+\phi_{i, j}$. If each agent has at most one friend (an important assumption for the sequel), in a generic instance no agent is indifferent between two plots, even if one of them is adjacent to her friend's plot. 

\section{Optimal Friend-Constrained Allocations}\label{sec:compl}
We first analyze the complexity of finding (approximately) optimal allocations under the assumption of complete information, i.e., when the weighted friendship graph as well as agents' valuation functions are known. Formally, given an instance $I$ of our problem and a positive rational value $T$, 
we ask whether there is an allocation $\A$ with $\SW(\A)\ge T$; we refer to this problem as {\sc SW-Opt}.

We first observe that even in the friendship-uniform case our 
problem is at least as hard as {\sc Subgraph Isomorphism}, and hence NP-hard \cite{gj79}. To see this, let all agents value all plots at $c\ge 0$ and each friendship at $\phi>0$; the maximum social welfare achievable is $n\times c + \phi\times |F|$.
This welfare is obtained in allocations in which every pair of friends receive adjacent plots; such allocations exist if and only if $\tup{N, F^*}$ is isomorphic to a subgraph of the plot graph $\G$. This observation establishes the following proposition.

\begin{proposition}\label{prop:np-hard}
{\sc SW-Opt} is NP-complete. This result
holds even if there exist $c, \phi>0$ such that
$u_i(v)=c$ for all $i\in N$, $v\in V$ and $\phi_{i, j}=\phi$ for all $(i, j)\in F$.
\end{proposition}

The proof of Proposition~\ref{prop:np-hard} shows that {\sc SW-Opt} is NP-hard even if the input instance is friendship-uniform and (i) $\tup{N, F^*}$ consists of a single clique and a collection of singletons (in which case our problem is at least as hard as {\sc Clique}), or (ii) $\tup{N, F^*}$ has maximum degree 2 (in which case our problem is at least as hard as {\sc Hamiltonian Cycle}). The reduction from {\sc Clique} with $c=0$ also shows that {\sc SW-Opt} is hard to approximate. 

Motivated by these hardness/inapproximability results, in the remainder of the paper we focus on the setting where $\tup{N, F^*}$ has maximum degree $1$, i.e., it is a collection of edges (pairs of friends) and singleton nodes. 
In this case, the respective subgraph isomorphism problem reduces to finding a maximum matching in the plot graph, which can be done in polynomial time. 
While at a first glance this variant of the model may appear to be very restrictive, it is quite natural in our setting. Indeed, buying land is a serious commitment, so the `friendships' in our context are typically sibling relationships, or other tight and long-running connections between households, and it is unlikely that a household is engaged in several such connections. 

Nevertheless, even this special case of {\sc SW-Opt} is NP-hard. The hardness result holds even in the friendship-uniform case and if the plot graph $\G$ consists of a single path (i.e., plots are located along a road) and several isolated plots. 

\begin{restatable}{theorem}{npcompPRM}\label{thm:sw-nphard}
{\sc SW-Opt} is NP-complete even if the instance is binary and friendship-uniform, the friendship graph $\tup{N, F^*}$ has maximum degree $1$, and the plot graph $\G$ consists of a single path and isolated nodes.
\end{restatable}
\begin{proof}
It is immediate that this problem is in NP: we can compute the social welfare of a given allocation using formula~\eqref{eqn:utility-fn}. 
To prove hardness, we provide an NP-hardness reduction from {\sc Path Rainbow Matching}.

An instance of the {\sc Path Rainbow Matching} problem is given by an integer $k$, and a properly edge-colored path, i.e., an undirected path $P=\tup{V, E}$ with vertices $V=\{v_1, \dots, v_s\}$, and edges $E=\{e_1, \dots, e_{s-1}\}$ such that $e_i=\{v_i, v_{i+1}\}$ for all $i=1, \dots, s-1$ together with
a finite set of colors $C = \{c_1,\dots ,c_q\}$, and a mapping $\xi: E\to C$ from edges of $P$ to colors such that $\xi(e_i)\neq \xi(e_{i+1})$ for each $i=1, \dots, s-1$. An instance $\tup{k, P, C, \xi}$ is a `yes'-instance if there exists a subset of edges $M\subseteq E$ with $|M|\ge k$ such that all edges in $M$ are pairwise disjoint and have different colors, and a `no'-instance otherwise. This problem is known to be NP-hard \cite{Le2014}.

Given an instance of {\sc Path Rainbow Matching} $\tup{ k, P, C, \xi}$, we construct an instance $\tup{N, \V, \E, F, (u_i)_{i\in N}, (\phi_{i, j})_{(i, j)\in F}}$ of our problem as follows. 

We set $N = \{1,\dots, 2q+s\}$: there are two agents corresponding to each color and $s$ additional dummy agents. We refer to agents $1, \dots, 2q$ as {\em color agents}.

The plot graph $(\V, \E)$ has $\V = \{v_1,\dots , v_{s}, w_1, \dots, w_{2q}\}$, and $\E = \{\{v_i,v_{i+1}\}: i=1,\dots,s-1 \}$, i.e., it is a copy of the given path instance together with $2q$ additional isolated plots (one for each color agent). 

The agents are friends if and only if they correspond to the same color and all friendships have weight $\phi=.1$: we set $F = \{ (2i-1,2i), (2i, 2i-1): 1\le i \le q\}$ and $\phi_{i, j} = .1$ for each $(i, j)\in F$. 

The agents' plot valuations are defined as follows. Dummy agents value all plots at $0$: $u_i(x)=0$ for each $i=2q+1, \dots, 2q+s$  and each $x\in \V$. Each color agent values `her' isolated plot at $1$ and all other isolated plots at $0$: for each $i=1, \dots, 2q$ we have $u_i(w_i)=1$, $u_i(w_j)=0$ for $j\neq i$. Also, for each edge $\{x, y\}$ with $\xi(\{x, y\})=c_i$ one agent in $\{2i-1, 2i\}$ values $x$ at $0$ and $y$ at $1$, and the other agent values $y$ at $0$ and $x$ at $1$. 
Specifically, for every color $c_i\in C$
let $\E_i= \{e\in \E: \xi(e)=c_i\}$, and suppose that  $\E_i=\{\{v_{i_1}, v_{j_1}\}, \dots, \{v_{i_r}, v_{j_r}\}\}$, where $i_1< i_2<\dots< i_r$ and $j_\ell=i_\ell+1$ for $\ell=1, \dots, r$. Then 
agent $2i-1$ values a plot $v_k\in\{v_1, \dots, v_s\}$ at $1$ if
$k=i_\ell$ for an odd value of $\ell$ or $k=j_\ell$ for an even value of $\ell$, and otherwise she values it at $0$.
Similarly, agent $2i$ values a plot $v_k\in\{v_1, \dots, v_s\}$ at $1$ if
if $v=i_\ell$ for an even value of $\ell$ or $v=j_\ell$ for an even value of $j$, and otherwise she values it at $0$.

We claim that our instance admits an assignment $\A$ with  $\SW(\A)\geq 2q+2k\phi$ if and only if $\tup{k, P, C, \xi}$ is a `yes'-instance of {\sc Path Rainbow Matching}. Indeed, suppose we start with a `yes'-instance of {\sc Path Rainbow Matching}, and let $M$ be the respective matching. For each edge $e=\{u, v\}\in M$, if $\xi(e)=c_i$, we assign color agents $2i-1$ and $2i$ to the endpoints of $e$ so that each of them is given the endpoint that she values at $1$; we assign the remaining color agents to their preferred isolated plots, while the dummy agents are matched arbitrarily to the remaining nodes. Then each color agent values her plot at $1$, and in addition there are $k$ pairs of friends who are allocated adjacent plots, so the overall social welfare is $2q+2k\phi$.

Conversely, suppose that there is an allocation $\A$ with $\SW(\A)\ge 2q+2k\phi$. Suppose first there exists some color agent $i$ with $U_i(\A)<1$. Then $\A(j)=w_i$ for some agent $j\in N\setminus\{i\}$, and $u_j(w_i)=0$, so we can swap $i$ and $j$ and increase the overall social welfare: even if $\A(i)$ is adjacent to the plot of $i$'s friend, the loss in social welfare caused by moving $i$ away from her friend is at most $2\phi<1$, and the gain in plot values is $1$. Thus, we can assume that in $\A$ each color agent values her plot at $1$. This means that there exist at least $k$ pairs of friends who are allocated adjacent plots, 
with each friend valuing her plot at $1$. Let $2i-1$, $2i$ be some such pair of friends, and suppose that they have been allocated plots $v_j, v_{j+1}$. Then either $\xi(\{v_j, v_{j+1}\})=c_i$ or 
$\xi(\{v_{j-1}, v_j\})=\xi(\{v_{j+1}, v_{j+2}\})=c_i$.
However, the latter case is impossible: we defined the plot valuation functions so that at least one of the agents $2i-1$ and $2i$ values both $v_j$ and $v_{j+1}$ at $0$ in this case.
Thus, these pairs correspond to a rainbow matching in $P$ of size at least $k$. 
 \end{proof}

On the positive side, if the friendship graph has maximum degree $1$, 
the problem of finding an allocation with maximum social welfare admits a poly-time $2$-approximation algorithm.

\begin{restatable}{theorem}{thmapproxoptimization}\label{thm:2-approx}
Given an instance $I$
where $\tup{N, F^*}$ has maximum degree 1, we can compute in polynomial time
an allocation $\A^*$ such that
$\SW(\A^*)\ge \frac12\OPT(I)$.
\end{restatable}
\begin{proof}
Our algorithm proceeds as follows. 

First, we need to find a maximum matching in the graph $\tup{\V, \E}$; let $\{\{v_1, w_1\}, \dots, \{v_s, w_s\}\}$ be the set of edges of this matching. Suppose that $F^*=\{\{i_1, j_1\}, \dots, \{i_t, j_t\}\}$, where 
$\phi_{i_\ell, j_\ell}+\phi_{j_\ell, i_\ell} \ge 
\phi_{i_r, j_r}+\phi_{j_r, i_r}
$
whenever $1\le \ell<r\le t$, i.e., the edges in $F^*$ are sorted by the total friendship weight in non-increasing order.
Then for each $k=1, \dots, \min\{t, s\}$ we allocate plot $v_k$ to $i_k$ and plot $w_k$ to $j_k$; all remaining plots are matched arbitrarily to the remaining agents. Let the resulting allocation be $\A_1$.

Second, we consider the weighted complete bipartite graph with parts $N$ and $\V$
where the weight of an edge $\{i, v\}\in N\times \V$ is equal to $u_i(v)$, and compute an allocation that corresponds to a maximum-weight matching in this graph; let this allocation be $\A_2$.

We output the better of the two allocations $\A_1$ and $\A_2$ (breaking ties arbitrarily). To see that this algorithm provides a $\frac12$-approxi\-ma\-tion, consider an arbitrary allocation $\A$. Under this allocation, at most $\min\{s, t\}$ pairs of friends are allocated adjacent plots, so the total utility they derive from friendship is at most $\SW(\A_1)$. Moreover, the total value that the agents assign to their plots under $\A$ is at most $\SW(\A_2)$. Thus, we have $\SW(\A)\le \SW(\A_1)+\SW(\A_2)$, whereas we output an allocation whose social welfare is at least $\frac12(\SW(\A_1)+\SW(\A_2))$. Moreover, both $\A_1$ and $\A_2$ can be computed in polynomial time, which concludes the proof.
\end{proof}

To summarize, for friendship graphs of maximum degree 1,
{\sc SW-Opt} is NP-hard, but admits a simple $2$-approximation algorithm. In the remainder of the paper, 
we restrict ourselves to friendship graphs of maximum degree 1, and ask if this constraint allows us to find good allocations even when agents' plot values and/or friendships are not publicly known.

\section{Plot Allocation Mechanisms}\label{sec:truth}

In this section we adopt a {\em mechanism design} perspective. That is, we are interested in deterministic/randomized mechanisms (without money) that elicit valuations and friendships, and output an allocation based on the reports; these mechanisms should be simple to describe and participate in, and produce good allocations, even when agents are strategic. 

We consider mechanisms where agents pick plots directly: the simplest such mechanism is the {\sc Serial Dictatorship} mechanism, where agents sequentially pick plots. 
Given that agents' utilities may depend on what other agents pick in subsequent rounds, we require that agents should be able to efficiently compute optimal strategies.
In addition, we require agents to report friendships. These reports are used to select the picking order, and to possibly restrict agents' plot choices.
Such mechanisms are easy to describe, making the allocation procedure transparent --- an important concern in our setting. 

Formally, we say that a deterministic mechanism is {\em friendship-truthful (FT)} if no agent can increase her utility by misreporting friendship information no matter what other agents report and no matter which plots they pick. A randomized mechanism is {\em universally FT} if it is friendship-truthful for every choice of its random bits, even when agents know the random bits used by the mechanism. 
A deterministic mechanism is {\em Pareto optimal (PO)} if it is guaranteed to output a PO allocation on every input; a randomized mechanism is  {\em universally PO} if it outputs a PO allocation on every input and for every choice of its random bits.

We are now ready to discuss mechanisms for land allocation with friends. We begin with serial dictatorship, 
identify its shortcomings, and explore several ways to overcome them. We derive a mechanism that is universally friendship-truthful, poly-time computable, and universally PO.

\subsection{Serial Dictatorship}\label{sec:rsd}
A natural starting point in our analysis is the {\sc (Random) Serial Dictatorship} (\rsd) mechanism \cite{abdulkadiroglu1998rsd,brandl2016rsd}. 
In the deterministic version of this mechanism, agents sequentially pick the plots, in a predetermined order; in the randomized version, agent order is chosen uniformly at random. 
For the one-sided matching problem, which is a special case of our problem, the optimal strategy of every agent under the \sd mechanism is simple: she should simply choose the best available plot. 
Moreover, for one-sided matching the (R)SD mechanism is (universally) PO as long as all agents have generic utilities. Its performance with respect to the social welfare is well-understood; in particular, for binary utilities, a variant of this mechanism offers a constant-factor approximation to the optimal social welfare \cite{adamczyk2014efficiency}. 

However, in the presence of friendships the agents' decision problem under \rsd becomes much more complicated, as illustrated by the following example.

\begin{example}\label{ex:sd}
Consider an instance with agents $1$, $2$, $3$, $4$, and plots $v_1, v_2, v_3, v_4$, arranged on a path. 
Let $F^*=\{\{1, 4\}, \{2, 3\}\}$. Suppose that agents' values for the plots are given by the table below and $\phi_{i, j}=.4$ for all $(i, j)\in F$. Consider what happens when we run the SD mechanism on this instance, with agent order $(1, 2, 3, 4)$.



\begin{center}
\renewcommand{\arraystretch}{0.9}
\small
  \begin{tabular}{  l  c  c  c  c }
    \toprule
            & $v_1$ & $v_2$ & $v_3$ & $v_4$ \\ 
    \midrule
    agent 1 & .5    & .3    &  0    &  0    \\ 
   
    agent 2 & 0     & .5    & .3    &  0    \\ 
    
    agent 3 & 0     & .7    &  0    & .5    \\ 
   
    agent 4 & 0     & .5    &  0    &  0    \\
    \bottomrule
  \end{tabular}

\end{center}


Agent 1 picks first. If he were to choose $v_1$, agent 2 would face the choice between $v_2$ and $v_3$ ($v_4$ is obviously less attractive). While she prefers $v_2$, she realizes that if she were to choose $v_2$, agent 3, who is her friend, would choose the non-adjacent plot $v_4$, so agent 2's utility would be $.5$. If agent 2 chooses $v_3$, agent 3 would pick $v_2$, so agent 2's utility would be $.3+\phi_{2, 3} = .3+.4=.7$. Therefore, agent 2 picks $v_3$; agent 3 picks $v_2$ next, and finally agent 4 picks $v_4$. Under this scenario, agent 1 ends up several plots away from his friend, so his utility is $.5$.

Now, suppose that agent 1 chooses $v_2$ instead. While $u_1(v_2)=.3< u_1(v_1)$, in this case agent 2 would pick $v_3$, agent 3 would pick $v_4$, and agent 4 ends up with $v_1$, i.e., right next to agent 1. 
Thus, agent 1's total utility from choosing $v_2$ is $.7$. As his utility from choosing $v_3$ or $v_4$ is at most $.4$, his best choice is $v_2$, and the resulting allocation $\A$ is given by $\A(1)=v_2$, $\A(2)=v_3$, $\A(3)=v_4$, $\A(4)=v_1$.
\end{example}

Example~\ref{ex:sd} illustrates interesting phenomena that arise when using the SD mechanism.
First, when deciding, agent 1 must reason about the decisions of agents who pick their plots after him. 
To choose optimally, he must know agents' plot values and friendships: indeed, if agents 3 and 4 had a low value for $v_2$ and high values for $v_3$ and $v_4$, he could safely pick $v_1$, as $v_2$ would remain available for his friend. 

Second, agent 2's decision depends on the order of agents who pick after him. If agent 4 picked immediately after agent 1, then agent 1 could safely pick $v_1$ and expect agent $2$ to pick the adjacent plot $v_2$.
Consequently, his decision is even more difficult if the order of agents is unknown. In particular, if agent order is chosen uniformly at random (i.e., using \rsd with hidden agent order), then, to evaluate the expected utility for each selection, he must consider all $3!=6$ scenarios corresponding to the permutations of the other agents.

Third, on this instance the SD mechanism produced an allocation that is not Pareto optimal: agents $1$ and $2$ would benefit from swapping their plots.

Thus, SD fails most of our criteria for a good mechanism. While it is simple to describe, the agents' decision problem is far from simple (in fact, the best upper bound on its computational complexity we could obtain is PSPACE). Further, agents must reason about other agents' utilities as well as their own, and the outcome may fail to be PO.

These difficulties mainly stem from the fact that whenever an agent $i$ has a friend $j$ that comes after her in the permutation, $i$ must predict $j$'s decision. More specifically, for each available plot, $i$ needs to know whether $j$ can and will pick an adjacent plot on her turn. 
Clearly, this task is much easier when $j$'s turn follows immediately after $i$: indeed, in our example, agent 2 had a much easier time making up her mind than agent 1. Thus, we will now explore variants of the SD mechanism that enable friends to choose consecutively.


\subsection{Choose-Together-SD (CT-SD) Mechanisms}\label{sec:ctsd}
Following the argument outlined at the end of Section~\ref{sec:rsd}, 
we consider a variant of \rsd where if $i$ and $j$ are friends, they appear consecutively in the permutation.

{\sc Online Choose-Together \rsd (On-CT-RSD)} is our first implementation of this idea: at each step, the mechanism picks one of the unallocated agents uniformly at random. The agent then picks a plot and may declare another unallocated agent as her friend; if so, then her friend is the next to choose a plot (but cannot declare another friend). Let $\V_i$ denote the set of plots that agent $i$ can select from on her turn. We say that a plot $v$ is a {\em singleton plot in $\V_i$} if it is not adjacent to any other plot in $\V_i$.
If an agent $i$ has a friend who has not selected a plot yet, and $i$ selects a singleton plot in $\V_i$, then she will not be placed next to her friend in the resulting allocation.

Suppose first that the friendship information is publicly available, i.e., agent $i$ can declare agent $j$ to be her friend if and only if $(i, j)\in F$. In this case, under \OnCTRSD the agents can compute their strategies in polynomial time.

\begin{restatable}{theorem}{thmONCTRSDPloy}
Suppose that agents cannot misreport friendship information. Then each agent can compute her optimal strategy in polynomial time.
To compute her strategy, each agent only needs to know her preferences and the preferences of her friend (if she has one).
\end{restatable}
\begin{proof}
Since agents cannot misreport friendships, their strategic decisions are limited to what plot to pick. Consider an agent $i\in N$. If $i$ has no friends, she should simply pick the plot with the highest value among the available plots. Now, suppose that $(i, j)\in F$. If $i$ picks after $j$, she can choose the plot that maximizes her utility, taking her friend's (known) location into account. Finally, if $i$ picks before $j$, she can consider all available plots, and, for each option, check whether $j$ would choose one of the adjacent plots at the next step; in this way, she can determine which plot would maximize her utility. Note that to make her decision, $i$ does not have to reason about the utilities of agents in $N\setminus\{i, j\}$.
\end{proof}

A further appealing feature of \OnCTRSD is that it is {\em ordinal}, in the sense that agents make their choice based on comparing plot values (accounting for additional value if a friend will be adjacent).
However, even if friendships are public, \OnCTRSD allocations are not necessarily PO; in fact, as Example \ref{ex:onctrsd-po} shows, \OnCTRSD may output an allocation that is dominated by a better allocation.

\begin{example}\label{ex:onctrsd-po}
\OnCTRSD may produce an allocation $\A$ where another allocation $\A'$ exists with $U_i(\A')>U_i(\A)$ for all $i\in N$, even if agents cannot misreport friendships.

Consider an instance with agents $1,2,3$ and plots $v_1, v_2, v_3$, where $\E=\{\{v_2, v_3\}\}$. Let $F^*=\{\{1, 2\}\}$. Agents' plot valuations are shown below, and $\phi_{1, 2}=\phi_{2, 1} = .5$. Suppose that the \OnCTRSD mechanism picks agent $1$ first, so the order is $(1, 2, 3)$.


\begin{center}
\renewcommand{\arraystretch}{0.9}
\small
  \begin{tabular}{  l  c  c  c  }
    \toprule
            & $v_1$ & $v_2$ & $v_3$\\
    \midrule
    agent 1 & 1     & .9    & 0    \\ 
    agent 2 & 1     & 0     & .4   \\ 
    agent 3 & 1     & .1    & 0    \\ 
    \bottomrule
  \end{tabular}
\end{center}


Agent 1 can guarantee herself a utility of $1$ by picking $v_1$. Her utility can be improved if she picked $v_2$ and her friend, agent 2, cooperates by picking $v_3$. However, agent $2$ would prefer $v_1$ if it is available. Hence, the plot $v_1$ remains agent 1's best choice, and the mechanism produces the allocation 
$\A(1)=v_1$, $\A(2)=v_3$, $\A(3)=v_2$. Now, an allocation $\A'$ given by $\A'(1)=v_2$, $\A'(2)=v_3$, $\A'(3)=v_1$ dominates $\A$ with $U_i(\A')>U_i(\A)$ for all $i\in N$.
\end{example}

Example~\ref{ex:onctrsd-po} fails to produce a PO allocation: agent 2 does not choose a plot adjacent to her friend's because she gains more from choosing $v_1$ over $v_3$ than she gains from friendship.
Indeed, if we change $\phi$ from $.5$ to $1$, \OnCTRSD produces a PO allocation. 
This observation can be generalized.

\begin{restatable}{theorem}{thmPOONCTRSD}\label{thm:PO-On-CT-RSD}
\OnCTRSD is universally PO on generic instances $\tup{N, \V, \E, F, (u_i)_{i\in N}, (\phi_{i, j})_{(i, j)\in F}}$ with $\phi_{\min} >1$.
\end{restatable}

\begin{proof}
Suppose for the sake of contradiction that, given an instance of our problem with $\phi_{\min}>1$, \OnCTRSD produces an allocation $\A$, yet there exist another allocation $\A'$ for this instance such that $U_\ell(\A')\ge U_\ell(\A)$ for all $\ell\in N$ and $U_i(\A')>U_i(\A)$ for some $i\in N$. We can assume without loss of generality that under \OnCTRSD the picking order is $(1, 2, \dots, n)$. Let $i$ be the first agent in this order such that $U_i(\A')>U_i(\A)$; note that, since our instance is generic, this means that $\A(\ell)=\A'(\ell)$ for all $\ell<i$ and hence $\A'(i)\in\V_i$. 

Suppose first that $i$ has no friends. Then under \OnCTRSD she picks the most valuable plot in $\V_i$ and $\A'(i)\in\V_i$, so we have $U_i(\A) = u_i(\A(i))\ge u_i(\A'(i)) = U_i(\A')$, a contradiction. Thus, we can assume that $(i, j)\in F$ for some $j\in N$.

Now, suppose that in our run of \OnCTRSD agent $i$ picks her plot after $j$, and hence $\A'(j)=\A(j)$. Then the utility that $i$ would obtain by picking $\A'(i)$ in the execution of \OnCTRSD is equal to the utility she obtains in $\A'$; since $\A'(i)\in\V_i$, we obtain a contradiction again.

It remains to consider the case where in our run of \OnCTRSD agent $i$ chooses before agent $j$ (and then $j$ chooses next). 
Then $i$'s best strategy is to pick the highest-value non-singleton plot in $\V_i$ (and to simply pick the highest-value plot if all plots in $\V_i$ are singletons). Indeed, if $i$ picks a non-singleton plot in $\V_i$, since $\phi_{j, i}>1$, agent $j$ would necessarily pick an adjacent plot in the next step, and, since $\phi_{i, j}>1$, agent $i$ would derive a higher utility from this choice than from any singleton plot in $\V_i$.

Suppose first that all plots in $\V_i$ are singletons, and hence under \OnCTRSD agent $i$ picks the highest-value singleton plot in $\V_i$. Then $\A'(i)$, too, is a singleton plot in $\V_i$, 
i.e., in $\A$ all plots adjacent to $\A'(i)$ are occupied by agents who appear before $i$ in the picking order, and we know that these agents are allocated the same plots in $\A'$. Thus, $i$ and $j$ are not allocated adjacent plots in $\A'$, and hence $U_i(\A')=u_i(\A'(i))\le u_i(\A(i))=U_i(\A)$, where the inequality holds since under \OnCTRSD agent $i$ picks the highest-value singleton plot in $\V_i$.
Thus, we obtain a contradiction in this case.

To complete the proof, suppose that $i$ picks a non-singleton plot in $\V_i$ under \OnCTRSD, and therefore $U_i(\A)>1$. Then it has to be the case that $U_i(\A')>1$, i.e., in $\A'$ agents $i$ and $j$ are allocated adjacent plots. Thus, $\A'(i)$ is a non-singleton plot in $\V_i$, but then we obtain a contradiction again, since $\A(i)$ is the highest-value non-singleton plot in $\V_i$, so $U_i(\A)=u_i(\A(i))+\phi_{i, j}\ge u_i(\A'(i))+\phi_{i, j}=U_i(\A')$. 
\end{proof}

So far we have assumed that agents cannot misreport their friendships. Let us now examine the role of this assumption.

\begin{proposition}
\OnCTRSD is not universally FT. 
\end{proposition}\label{prop:on-ct-rsd-not-friend}
\begin{proof}
Let us revisit Example~\ref{ex:onctrsd-po}. 
Suppose again that agent 1 is the first in the picking order. We argued that if agent 1 declares agent 2 as her friend, she maximizes her utility by picking the plot $v_1$, resulting in a total utility of $1$. Suppose, however, that agent 1 picks plot $v_2$ and declares agent 3 to be her friend. Then agent 3 chooses next, and picks the plot $v_1$. Agent 2 is then forced to pick plot $v_3$, so that the total utility of agent 1 is $u_1(v_2)+\phi_{1, 2}=1.4>1$. Thus, agent 1 benefits from misreporting friendship information.
\end{proof}

However, as is the case for PO, if $\phi_{\min}>1$, this negative result no longer holds.

\begin{restatable}{theorem}{thmOnCTRSDTruth}\label{thm:OnCTRSD-ft}
\OnCTRSD is universally friendship-truthful for every instance with $\phi_{\min}>1$.
\end{restatable}
\begin{proof}
Clearly, if an agent has no friends, she cannot benefit from declaring another agent to be her friend, as it would not give her access to a better plot. Similarly, if an agent $j$ is `invited' by $i$, i.e., $j$ picks right after $i$ because $i$ declared $j$ to be her friend, $j$ is not asked to report her friendship information, so she has no opportunity to misreport. Now, suppose that $i$ has a friend (say, $j$), and $i$ gets to pick a plot before $j$.
If all plots in $\V_i$ are singletons, then friendship information is irrelevant, and $i$ has no incentive to misreport. 
Otherwise, let $v$ be a highest-value non-singleton plot in $\V_i$. Then the highest utility $i$ can hope to get in this run of the mechanism is $u_i(v)+\phi_{i, j}$, which is exactly the utility she would get by picking $v$ and declaring $j$ to be her friend: indeed, since $\phi_{j, i}>1$, $j$ would then choose a plot adjacent to $v$. Hence, $i$ has no incentive to misreport the friendship information in this case as well. 
\end{proof}

To summarize, \OnCTRSD is an attractive mechanism if $\phi_{\min}>1$; however, in general it is neither universally PO nor universally friendship-truthful. 
We next discuss modifying this mechanism to avoid these issues.


\subsection{Choose-Adjacent-SD (CA-SD) Mechanisms}
The main reason why \OnCTRSD fails both PO and friendship-truthfulness when $\phi_{\min}<1$
is that when agent $i$ declares agent $j$ to be her friend, $j$ can `jump the queue', but may choose a plot not adjacent to $i$'s. 
We now consider a mechanism that explicitly prohibits such behavior.

Specifically, this mechanism, {\sc Online Choose-Adjacent \rsd} (\OnCARSD),  proceeds identically to \OnCTRSD with one difference: if agent $i$ declares $j$ to be her friend and chooses a non-singleton plot in $\V_i$, at the next step $j$ \emph{must} choose a plot adjacent to $i$'s; if $i$ chooses a singleton plot in $\V_i$, $j$ can then choose any plot in $\V_j$. 
Alternatively, if an agent selects a singleton plot, the mechanism may forbid her from declaring a friend; this has no impact on our analysis. 

Note that \OnCARSD is equivalent to \OnCTRSD if $\phi_{\min}>1$: whenever an agent $i$ chooses after her friend, she would pick an adjacent plot if at all possible. 
However, in general, the mechanisms are different: e.g., on the instance described in Example~\ref{ex:onctrsd-po} \OnCARSD would output an allocation $\A$ with $\A(1)=v_2$, $\A(2)=v_3$, $\A(3)=v_1$.

It turns out that \OnCARSD satisfies the criteria formulated in the beginning of this section.

\begin{restatable}{theorem}{thmPOCARSD}\label{thm:PO-ON-CA-RSD}
\OnCARSD is universally PO and universally friendship-truthful; moreover, agents' strategies are polynomial-time computable.
\end{restatable}

\begin{proof}
The analysis is similar to the analysis for \OnCTRSD with $\phi>1$. Suppose for the sake of contradiction that, given an instance of our problem, \OnCARSD produces an allocation $\A$, yet there exist another allocation $\A'$ for this instance such that $U_\ell(\A')\ge U_\ell(\A)$ for all $\ell\in N$ and $U_i(\A')>U_i(\A)$ for some $i\in N$. We can assume without loss of generality that under \OnCARSD the picking order is $(1, 2, \dots, n)$. Let $i$ be the first agent in this order such that $U_i(\A')>U_i(\A)$; since or instance is generic, this means that $\A(\ell)=\A'(\ell)$ for all $\ell<i$ and hence $\A'(i)\in\V_i$. 

If $i$ has no friends, then under \OnCARSD she picks the most valuable plot in $\V_i$; as $\A'(i)\in\V_i$, we have $U_i(\A) = u_i(\A(i))\ge u_i(\A'(i)) = U_i(\A')$, a contradiction. 

Now, suppose that $(i, j)\in F$ for some $j\in N$.
If in our run of \OnCARSD agent $i$ picks her plot after $j$, we have $\A'(j)=\A(j)$. Suppose first that agent $j$ picked a singleton plot in $\V_j$, so that the choice of agent $i$ is unconstrained, and hence she picks the most valuable plot in $\V_i$. Then the analysis is similar to the previous case: as $\A'(i)\in\V_i$, we have $U_i(\A) = u_i(\A(i))\ge u_i(\A'(i)) = U_i(\A')$, a contradiction. On the other hand, if $j$ did not pick a singleton plot, then $\A(i)$ is the most valuable plot among the plots that are adjacent to $\A(j)$. Thus, if $\A'(i)\neq \A(i)$ then $\A'(i)$ is not adjacent to $\A'(j)=\A(j)$. But this means that $U_j(\A')=U_j(\A)-\phi_{j, i}$, a contradiction with our assumption that $U_\ell(\A')\ge U_\ell(\A)$ for all $\ell\in N$.

Finally, suppose that in our run of \OnCARSD agent $i$ picks her plot before $j$. Note that $\A'(i), \A'(j)\in V_i$. If these plots are adjacent, then $i$ can pick $\A'(i)$; as $j$ will be forced to pick an adjacent plot in the next iteration, we have $U_i(\A)\ge U_i(\A')$. Otherwise, we have $U_i(\A')=u_i(\A'(i))$, so agent $i$ can obtain the same utility as in $\A'$ simply by choosing $\A'(i)$.
This completes the proof of Pareto optimality.

For friendship truthfulness, the proof is very similar to the proof of Theorem~\ref{thm:OnCTRSD-ft}: just as in that proof, an agent does not benefit from misreporting if she does not have a friend or if she chooses after her friend. Further, if agent $i$ chooses before her friend $j$, the maximum utility she can obtain is the higher of $\max_{v\in\V_i}u_i(v)$ and $\max_{v\in \V^{\textit{ns}}_i}u_i(v)+\phi_{i, j}$, where $\V^\textit{ns}_i$ is the set of non-singleton plots in $\V_i$, and she can guarantee herself that utility by reporting truthfully.

Finally, the polynomial-time computability follows from the description of the optimal strategies given in the previous paragraph.
\end{proof}

\OnCARSD has many attractive properties: it is simple, agents can compute their strategies efficiently and without knowing other agents' preferences (not even their friends' preferences!), and the mechanism always produces a PO allocation. 
However, if agents' value for being close to their friends is low relative to the differences among the plot values, they may find this mechanism to be highly problematic.
\begin{example}\label{ex:nasty}
Let $\G$ consist of a single edge $\{v, w\}$ and $n-2$ isolated plots. Every agent values $w$ at $0$ and all other plots at $1$. Suppose all friendships have value $\phi=.1$. If agents $i$ and $j$ are friends and $i$ is the first agent to pick, then $i$ will choose $v$ (as she can then benefit from being next to $j$) and $j$ will be forced to choose $w$ and get the worst plot in $\V$.
\end{example}

One may then wonder if it is possible to modify \OnCARSD to give an agent the option to decline her friend's `invitation' and choose at a later point, but without having her plot choices constrained. 
There are several ways to implement this idea. 
For instance, if agent $i$ declares a remaining agent $j$ as a friend, we can offer $j$ the choice of (1) picking a plot right after $i$, but it must be adjacent to $i$'s plot (if at all possible), or (2) declining the invitation and returning to the pool of remaining agents; we refer to this mechanism as {\sc CA-Back-To-Pool-RSD} (\CABPRSD).
Alternatively, we can sample a default agent order in advance (uniformly among all possible $n!$ orders), announce it to all agents, and then approach the agents one by one in this order, asking them to pick a plot and to declare a friend. If $i$ declares $j$ to be her friend, then $j$ can either accept the invitation, jump the queue and pick a plot adjacent to $i$'s (if such a plot exists); or, decline and keep her place in the queue (or, even more drastically, move to the end of the queue); we refer to these mechanisms as {\sc CA-Back-to-Queue-RSD} (\CABQRSD) and {\sc CA-Back-to-End-RSD} (\CABERSD), respectively.
These mechanisms seem to preserve the spirit of \OnCTRSD, but offer agents more flexibility. Unfortunately, our next example shows that neither is universally friendship-truthful.

\begin{example}\label{ex:lessnasty}
Consider an instance with agents $1, 2, 3, 4$, and plots $v_1, v_2, v_3, v_4$, arranged on a path in that order. 
Let $F^*=\{\{1, 4\}\}$. Suppose that agents' values for the plots are given by the table below and $\phi_{1, 4}=\phi_{4, 1} = .2$. 


\begin{center}
\renewcommand{\arraystretch}{0.9}
\small
  \begin{tabular}{  l  c  c  c  c }
    \toprule
            & $v_1$ & $v_2$ & $v_3$ & $v_4$\\ 
    \midrule
    agent 1 & 0     & 1     &  0   &  0  \\ 
    agent 2 & .3    & 0     & .1   &  .2  \\ 
    agent 3 & .3    & 0     & .2   &  0  \\ 
    agent 4 & 0     & 0     &  0   &  1  \\
    \bottomrule
  \end{tabular}

\end{center}

Under \OnCARSD, if agent 1 picks first, she would pick $v_2$, and announce agent 4 as her friend, forcing agent 4 to pick an adjacent plot. 
Under \CABPRSD
agent 4 can decline this option, in which case agents 2, 3, and 4 pick their plots in random order. Agent 4 chooses next w.p. $1/3$, in which case she will be able to pick her favorite plot. Thus, her expected utility is at least $1/3>\phi_{4, 1}$, so she will not confirm friendship with agent 1. Thus, under \CABPRSD, if agent 1 declares agent 4 as her friend, her utility is 1. 

Now, suppose agent 1 falsely declares agent 3 as her friend. Agent 3 has no reason to decline this invitation; indeed, accepting ensures that she receives her favorite plot (rather than risk losing it to agent 2). Thus, agent 3 accepts and picks $v_1$. Agents 2 and 4 prefer $v_4$ to $v_3$, so the first to pick claims $v_4$ for themselves. Thus, with probability $.5$ agent 4 ends up with $v_3$, which is adjacent to agent 1's plot. Hence, 
under \CABPRSD, agent 1's expected utility from declaring agent 3 as her friend is $1+.5\times .2=1.1$, which is higher than her utility from a truthful declaration.

The same argument shows that \CABQRSD and \CABERSD are not friendship-truthful: if the order is $(1, 2, 3, 4)$, then agent 1 prefers declaring agent 3 as her friend.
\end{example}
Thus, there does not seem to be an easy way to make \OnCARSD more flexible while retaining universal PO and friendship-truthfulness. 


\section{Social Welfare Maximization}\label{sec:swm}
So far, we focused on simplicity, polynomial-time computability and friendship-truthfulness; the only allocative efficiency measure we discussed was PO, which is a relatively weak requirement.
We will now derive bounds on the social welfare of the assignments produced by \OnCTRSD and \OnCARSD and their variants. 
For simplicity, we focus on friendship-uniform instances, i.e., we assume that $\phi_{i, j}=\phi$ for some fixed $\phi$ and all $(i, j)\in F$. 
Since our problem is at least as hard as the one-sided matching problem, we cannot expect \rsd and its variants to perform well for general valuations; thus, we focus on binary instances. 

For binary utilities, \citet{adamczyk2014efficiency} propose the following modification of the \rsd mechanism, which we call \rsdstar. In each iteration, before picking the next agent, \rsdstar asks all remaining agents to report if they have a positive value for some available plot. If some agents answer `yes', \rsdstar picks one of them uniformly at random, lets her pick a plot, and starts the next iteration; otherwise, \rsdstar arbitrarily pairs remaining agents with remaining plots and terminates. 
\rsdstar reduces waste while maintaining truthfulness, 
giving a $1.45$-approximation to the optimal social welfare under binary valuations; can we obtain a similar approximation ratio in our setting?

Our first result is discouraging: \OnCTRSD may produce assignments with very poor social welfare, even if $\phi>1$, i.e., even in the setting where it is PO for generic instances and friendship-truthful.

\begin{example}\label{ex:on-ct-rsd-sw}
Consider an instance with $N=\{1, \dots, n\}$, $\V=\{v_1, \dots, v_n\}$, where $\E=\{\{v_1, v_2\}\}$.
Suppose that $F^*=\{\{1, 2\}\}$, $\phi_{1, 2}=\phi_{2, 1} = 100$. All agents 
value $v_1$ at $1$ and all other plots at $0$.

Under \OnCTRSD agents 1 and 2 end up in adjacent plots if and only if one of then appears first in the picking order, i.e., with probability $\frac{2}{n}$. Thus, the expected social welfare under this mechanism is $1+2\times\frac{2}{n}\times 100$, whereas the optimal social welfare is $202$.
\end{example}

When friendships are valuable, i.e., $\phi\gg 1$, we would like to avoid the situation described in Example~\ref{ex:on-ct-rsd-sw}. 
This can be accomplished by prioritizing pairs of friends, i.e., ensuring that pairs of friends choose first, followed by agents who do not have friends. 
This requires us to elicit friendship information offline, before agents start picking plots. 
As we cannot assume that agents will report this information truthfully, to fully specify such a mechanism, we need to handle inconsistent reports: what if $i$ says that $j$ is her friend, but $j$ does not say that $i$ is her friend?
We take the conservative approach and treat $i$ and $j$ as friends iff both declare this friendship. 

Formally, this mechanism, {\sc Friends-First Choose-Together RSD$^*$}  (\FFCTRSDstar) proceeds as follows. First, each agent reports who their friend is (or $\varnothing$ for no friends). Let 
$P$ be the set of pairs $\{i,j\}$ who report each other as friends.
We pick agents in the following order: as long as there exist a pair of adjacent unoccupied plots and $P\neq\varnothing$, we randomly remove a pair of agents $\{i, j\}$ from $P$; $i$ and $j$ then choose their plots (in random order). We execute \rsdstar over remaining agents and plots once $P = \emptyset$ or no adjacent plots are available.
We analyze the performance of \FFCTRSDstar, under the assumption that agents cannot lie about their friendships and $\phi>1$.

\begin{restatable}{theorem}{thmSWFF}\label{thm:ff-scw}
Let $\A$ be the output of \FFCTRSDstar on a binary instance $I$ with $\phi_{\min}>1$, where agents truthfully report friendships. Then ${\mathbb{E}}(\SW(\A))\ge \frac{1}{4}\OPT(I)$.
\end{restatable}


Of course, since \FFCTRSDstar prioritizes pairs of friends, we cannot expect it to be friendship-truthful.
Thus, if friendship-truthfulness is considered desirable, we are left with \OnCARSD or its variants. Specifically, \OnCARSD, too, can be modified by pushing friendless agents who value all available plots at $0$ to the back of the queue, in the spirit of \rsdstar;
we refer to this mechanism as \OnCARSDstar. It can be verified that this mechanism remains friendship-truthful. 

Since \OnCARSDstar does not prioritize friendships, we cannot expect it to have a constant approximation ratio (consider, e.g., its performance on the instance in Example~\ref{ex:on-ct-rsd-sw}).
However, if $\phi>1$, we can bound the approximation ratio of \OnCARSDstar in terms of~$\phi$.

\begin{restatable}{theorem}{thmSWCASWC}\label{thm:on-ca-scw}
Let $\A$ be the output of \OnCARSDstar on a binary instance $I$ with $\phi_{\min}>1$. Then ${\mathbb{E}}(\SW(\A))\ge \frac{1}{2\phi+2}\OPT(I)$.
\end{restatable}
The positive results presented so far in this section are for the case $\phi>1$. For $\phi<1$, positive results are more elusive. In particular, it is no longer the case that \FFCTRSDstar has a constant approximation ratio.

\begin{restatable}{proposition}{propffswsmallphi}\label{prop:ff-scw-smallphi}
There exists a friendship-uniform binary instance $I$ with $\OPT(I)=2+2\phi$ such that 
the expected social welfare of the output of \FFCTRSD is at most $\frac{6}{n}+4\phi$.
\end{restatable}
\begin{proof}
Consider an instance with $N=\{1, \dots, n\}$, where $n=2k$ is even, 
$\V=\{v_1, \dots, v_k, w_1, \dots, w_k\}$, $\E=\{\{v_1, v_i\}: 2\le i\le k\}\cup\{\{w_1, w_i\}: 2\le i\le k\}\cup\{\{v_1, w_1\}\}$,
$F^*=\{\{2i-1, 2i\}: i=1, \dots, k\}$.
Suppose that $u_1(v_1)=u_1(w_1)=u_2(v_1)=u_2(w_2)=1$ and all other plot values are $0$.

If $\phi<.5$, an optimal allocation assigns $v_1$ and $w_1$ to agents $1$ and $2$, so that the social welfare is $2+2\phi$. Now, under \FFCTRSDstar the probability that agents 1 and 2 appear in the first two positions of the picking order is $\frac{2}{n}$, and the probability that they appear in the next two positions of the picking order is $\frac{2}{n}$ as well; if neither of these events happens, plots $v_1$ and $w_1$ will be occupied by agents who value them at $0$ (but derive positive utility from being next to their friend), so the social welfare will be at most $4\phi$. Thus, the expected social welfare of the allocation produced by \FFCTRSDstar is at most 
$2\times\frac{2}{n}+
1\times\frac{2}{n}+4\phi$.
\end{proof}

Our last result applies not just to variants of the \rsd mechanism, but to all truthful mechanisms: the approximation ratio of any such mechanism is at most $1+\frac{1}{2\phi}$, even if agents cannot misreport their friendship information.

\begin{restatable}{proposition}{proplowerbound}\label{prop:truthful}
Consider a mechanism $\mathcal M$ that has access to the friendship graph $\tup{N, F}$, asks the agents to report their values for the plots, and outputs an allocation based on the agents' report and the friendship graph. If no agent can benefit from misreporting her plot values under $\mathcal M$ then here exists a friendship-uniform binary instance $I=\tup{N, \V, F, (u_i)_{i\in N}, (\phi_{i, j})_{(i, j)\in F}}$ such that for the allocation $\A$ output by $\mathcal M$ we have $\frac{{\mathbb E}(\SW(\A))}{\OPT(I)}\le \frac{2\phi}{2\phi+1}$.
\end{restatable}

\begin{proof}
Let $I_1=\tup{N, \V, F, (u_i)_{i\in N}, (\phi_{i, j})_{(i, j)\in F}}$,
where $N=\{1, \dots, n\}$ and $n$ is even, $n=2k$, $\V=\{v_1, \dots, v_n\}$, $\E=\{\{v_1, v_i\}: 2\le i\le n\}$, $F=\{\{2i-1, 2i\}: 1\le i\le k\}$, 
$u_i(v)=0$ for all $i\in N$ and all $v\in \V$, and there exists a positive value $\phi$ such that $\phi_{i, j}=\phi$ for all $(i, j)\in F$. That is, the plot graph is a star with center $v_1$, and each agent has a friend and values all plots at $0$. We have $\SW(\A)=2\phi$ for every $\A\in\AA(I_1)$. 

By the pigeonhole principle, there exists a pair of friends $\{2i-1, 2i\}$ such that mechanism $\mathcal M$ allocates $v_1$ to $2i-1$ or $2i$ with probability at most $\frac{2}{n}$. Now, consider the instance $I_2$ that is obtained from $I_1$ by changing $u_{2i}(v_1)$ to $1$. Since $\mathcal M$ is truthful, agent $2i$
cannot increase her utility in $I_1$ by misreporting her utility function, so given $I_2$, $\mathcal M$ allocates $v_1$ to $2i$ with probability at most $\frac{2}{n}$. Thus, the expected social welfare of the allocation produced by $\mathcal M$ on $I_2$ is at most $\frac{2}{n}+2\phi$, whereas $\OPT(I_2)=1+2\phi$. As $n$ can be arbitrarily large, the bound follows.
\end{proof}


\section{Conclusions and Future Work}\label{sec:conclusions}
We have analyzed the problem of allocating plots of land to buyers who have intrinsic preferences over their neighbors. While the problem in its full generality offers several non-trivial computational challenges, we show that under some realistic assumptions on buyer preferences and permitted reports, it is possible to design simple mechanisms that maintain both truthful reporting and social welfare guarantees. 

We obtain positive results if all agents value their friendships highly ($\phi_{i, j}>1$), and even stronger positive results are known in the absence of friendships (i.e., if $\phi_{i, j}=0$). However, paradoxically, the presence of low-valued friendships may result in significant welfare loss, as shown by Proposition~\ref{prop:truthful}. To see why this may be the case, note that even low-value friendships may distort agents' behavior under RSD, thereby changing the allocation significantly.

We focused on RSD-like mechanisms for our problem; however, it may also be useful to consider other approaches. E.g., we can explore market-like mechanisms, where agents are allocated identical budgets and need to bid on plots and possibly on friendships, 
in the spirit of \citet{budish2011}.


\bibliographystyle{ACM-Reference-Format}  
\bibliography{ijcai20}  

\appendix

\section{Omitted Proofs from Section~\ref{sec:swm}}

\thmSWFF*
\begin{proof}
Consider the \FFCTRSD* at iteration $t$; we define $N^t$ be the left (unassigned) agents at time $t$, $\mathcal G^t$ be the topology unallocated plots at time $t$ with vertex set $\mathcal V^t$ and edges $\mathcal E^t$, $F^t$ be the singular friendship structure among $N^t$, $ \OPT^t$ be the optimal allocation of $N^t$ to the plot topology $\mathcal G^t$ and $\SW^t$ be the total social welfare by the algorithm by time $t$. We denote the past history of plot assignment before the $t$-th iteration as $\mathcal H^t$. We note that the $\SW(\mathcal A) = \SW^T$ ($T$ be the total number of iterations) and $ \SW(\OPT) = \SW(\OPT^0)$.
Now, at each iteration $t$, we divide all friendship pairs into four sets as follows: 
\begin{align*}
    &Z_1^t = \{ (i,j)\in F^t: u_l(\OPT^t) = 1 + \phi; l=i,j  \}\\
    &Z_2^t = \{ (i,j)\in F^t: u_i(\OPT^t) + u_j(\OPT^t) = 2\phi +1   \}\\
    &Z_3^t = \{ (i,j)\in F^t: u_l(\OPT^t) = \phi; l=i,j   \}\\
    &Z_4^t = \{ (i,j)\in F^t: u_l(\OPT^t) \leq 1; l=i,j   \}
\end{align*}
For convenience, we write $|Z_i^t| = z_i^t$ for all $i=1,\dots, 4$ and $|F^t| = f^t$. As friendship value $\phi >1$, at iteration $t$, if the first while loop runs then the randomly chosen pair will always be able to capture two adjacent plots due to higher friendship value $\phi$ than plot values and availability of adjacent plots. which implies that 
\begin{equation}\label{eqn:firstwhileSW}
    \Exp [\SW^{t+1} | \mathcal{H}^t ] \geq \SW^{t} + \frac{z_1^t}{f^t}(2\phi + 2) + \frac{z_2^t}{f^t}(2\phi + 1) + \frac{z_3^t + z_4^t}{f^t}(2\phi)  
\end{equation}

Now, we analyze the decrease in the optimal welfare after each iteration when the first while loop condition satisfies ($\mathcal E^t \neq \emptyset$ and $F^t \neq \emptyset$). Consider randomly selected pair $(i,j)$ at iteration $t$, they will grab any available adjacent plots among available plots $(v,w) \in \mathcal E^t$ due to higher friendship value than plot values. If randomly chosen $(i,j) \in Z_1^t \cup Z_2^t \cup Z_3^t $, then $(i,j)$ can destroy at most $2$ allocated adjacent plots to friends in $\OPT^t$ allocation and their own friendship. Therefor, it can cost at most $6\phi$ in friendship value of the $\OPT^t$ allocation and if $(i,j) \in Z_4^t$ then it can cost at most $4\phi$ in friendship value of the $\OPT^t$ allocation. Moreover, in the case when $(i,j)\in Z_1^t \cup Z_2^t \cup Z_4^t$, by grabbing adjacent plots $(v,w)$, it can disturb at most $4$ assigned plot (with value $1)$ in $\OPT^t$ allocation and when $(i,j)\in Z_3^t$, by grabbing adjacent plots $(v,w)$, it can disturb at most $2$ plot values in $\OPT^t$ allocation. Therefore we can write using equation~\ref{eqn:firstwhileSW}; 

\begin{align*}
    &\Exp [\SW(\OPT^{t+1}) - \SW(\OPT^{t}) | \mathcal{H}^t ]\\
    &\leq \frac{z_1^t+z_2^t}{f^t}(6\phi + 4) + \frac{z_3^t}{f^t}(6\phi +2) + \frac{z_4^t}{f^t}(4\phi + 4)\\
    &\leq \frac{4(2\phi + 2)z_1^t}{f^t} + \frac{4(2\phi + 1)z_2^t}{f^t} + \frac{8\phi( z_3^t + z_4^t)}{f^t}\\
    &\leq 4 \cdot \Exp [\SW^{t+1} - \SW^{t}  | \mathcal{H}^t ]
\end{align*}

Once all pairs of friends or adjacent edges in the topology are exhausted ($F^t = \emptyset$ or $\mathcal E^t = \emptyset$), our algorithm becomes random serial dictatorship in one-sided matching market. Theorem 2 in \cite{adamczyk2014efficiency} implies that

\begin{align*}
   &\Exp [\SW(\OPT^{t+1}) - \SW(\OPT^{t}) | \mathcal{H}^t ] \leq \\ 
   & 3 \cdot \Exp [\SW^{t+1} - \SW^{t}  | \mathcal{H}^t ] \leq 4 \cdot \Exp [\SW^{t+1} - \SW^{t}  | \mathcal{H}^t ]
\end{align*}

This implies that the sequence of random variables $X^0 = 0$ and $X^t - X^{t-1} = 4\cdot(\SW^t - \SW^{t-1}) - (\SW(\OPT)^t - \SW(\OPT)^{t-1})$ is a sub-martingale as $\Exp[X^{t+1}|\mathcal H^t] \geq X^t$. Therefore by Doobs Stopping Theorem, we get $\Exp[X^{T}|\mathcal H^t] \geq 0$ which implies,

\begin{align*}
    0 &\leq \Exp \bigg [ \sum_{k = 1}^T (X^k - X^{k-1} ) \bigg ] = 4 \cdot \Exp \bigg[ \sum_{k=1}^T (\SW^{k} - \SW^{k-1}) \bigg ]\\ 
    & - \Exp \bigg[ \sum_{k=1}^T (\SW(\OPT^{k-1}) -  \SW(\OPT^{k}) ) \bigg ]\\
    &\Rightarrow \Exp [ \SW(\mathcal{A}) ] \geq \frac{\OPT}{4}
\end{align*}
This concludes the proof. 
\end{proof}

\thmSWCASWC*

\begin{proof}
Let $N^t, \mathcal G^t, \mathcal{V}^t, \mathcal E^t, \mathcal H^t, \SW^t$ and $\OPT ^t$ are defined similar to the proof of Theorem~\ref{thm:ff-scw} at $t-$th iteration of \OnCARSD*. We denote $|N^t|$ as $n^t$ for simplification. We define the subsets of the $N^t$ as follows:
\begin{align*}
    &Z_1^t = \{ i\in N^t: i \text{ has friend $j$, } u_l(\OPT^t) = \phi + 1; l=i,j  \}\\
    &Z_2^t = \{ i\in N^t: i \text{ has friend $j$, } u_i(\OPT^t)+ u_j(\OPT^t) = 1 + 2\phi\}\\
    &Z_3^t = \{ i\in N^t: i \text{ has friend $j$, } u_l(\OPT^t) = \phi ; l=i,j  \}\\
    &Z_4^t = \{ i\in N^t: i \text{ has no friend, } u_l(\OPT^t) \leq 1; l=1,2 \}\\
    &Z_5^t = \{ i\in N^t: i \text{ has no friend, } u_i(v) = 0, \forall v\in \mathcal V^t  \}   
\end{align*}

Consider $t$ with $\mathcal E^t \neq \emptyset$, If randomly chosen $i\in Z_1^t\cup Z_2^t\cup Z_3^t\cup Z_4^t$, as value of friendship $\phi>1$ and $\mathcal E^t \neq \emptyset$, $i$ will declare her friend $j$ and pick a plot $v\in \mathcal V^t$ with available adjacent plot with maximizing her own utility then if possible her friends utility. If $i\in N^t \setminus (Z_1^t \cup Z^t_2 \cup Z_3^t\cup Z_4^t\cup Z_5^t)$, $i$ will pick a plot with value $1$. We note that the \OnCARSD never picks $i\in Z_5^t$. This implies 
\begin{equation}
\begin{split}
     \Exp [\SW^{t+1}| \mathcal H^t] \geq \SW^{t} + \frac{z_1^t}{n^t}(2\phi +2) + \frac{z_2^t}{2}(2\phi + 1)\\ +\frac{(z_3^t + z_4^t)2\phi}{n^t}
    + 1 - \frac{(z_1^t + z_2^t + z_3^t+ z_4^t + z_5^t)}{n^t}   
\end{split}
\end{equation}
We now analyse the decrease in optimal welfare at iteration $t$ whenever $\mathcal E^t \neq \phi $. If the randomly chosen agent $i$ at iteration $t$ belongs to $Z_1^t \cup Z_2^t \cup Z_3^t$ then the agent $i$ will pick the plot $v \in \mathcal V^t$ and force her friend $j$ to pick $w \in \mathcal N_v$. Therefore in any case, by grabbing two adjacent plots $(v,w) \in \mathcal E^t$, they can destroy at most $2$ other friendship values, their own friendship value in the $\OPT^t$ allocation where if $i\in Z_4^t$ then it can destroy at most $2$ friendship values in $\OPT^t$ allocation (as they are not assigned adjacent plots in $\OPT^t$ allocation). Now, we analyse the decrease in plot values in optimal welfare. If $i \in Z_1^t \cup Z_4^t$, they can destroy at most $4$ plot values in $\OPT^t$ allocation. Similarly, if $i \in Z_3^t \cup (N^t \setminus (Z_1^t \cup Z^t_2 \cup Z_3^t\cup Z_4^t\cup Z_5^t))$, it can destroy at most $2$ plot values in $\OPT^t$ allocation and for $i \in Z_2^t$ can destroy at most $3$ plot values in $\OPT^t$ allocation (only one of them assigned to high valued plot in $\OPT^t$). This implies that: 

\begin{align*}
    &\Exp [\SW(\OPT^{t}) - \SW(\OPT^{t+1}) | \mathcal{H}^t ] \\
    & \leq  \frac{z_1^t(6\phi + 4)}{n^t} + \frac{z_2^t(6\phi + 3)}{n^t} + \frac{z_3^t(6\phi + 2)}{n^t} + \frac{z_4^t(4\phi + 4)}{n^t}\\
    & \qquad \qquad + \Big( 1- \frac{ z_1^t + z_2^t + z_3^t + z_4^t + z_5^t}{n^t}\Big)(2\phi +2)  \\  
    & \leq  \frac{z_1^t(2\phi + 2)^2}{n^t} + \frac{z_2^t(2\phi + 1)(2\phi + 2)}{n^t} + \frac{z_3^t(2\phi)(2\phi + 2)}{n^t} +  \\
    & \frac{z_4^t(2\phi)(2\phi + 2)}{n^t} + \Big( 1- \frac{ z_1^t + z_2^t + z_3^t + z_4^t + z_5^t}{n^t}\Big)(2\phi +2)  \\
    &\leq (2\phi + 2) \cdot \Exp [\SW^{t+1} - \SW^{t}  | \mathcal{H}^t ]
\end{align*}

Note that $z_1^t+z_2^t + z_3^t +z_4^t + z_5^t \leq n^t$. Once all adjacent edges in the topology are exhausted ($\mathcal E^t = \emptyset$), our algorithm becomes a random serial dictatorship in a one-sided matching market. Theorem 2 in \cite{adamczyk2014efficiency} implies that

\begin{align*}
   &\Exp [\SW(\OPT^{t+1}) - \SW(\OPT^{t}) | \mathcal{H}^t ] \leq \\ 
   & 3 \cdot \Exp [\SW^{t+1} - \SW^{t}  | \mathcal{H}^t ] \leq (2\phi+2)  \cdot \Exp [\SW^{t+1} - \SW^{t}  | \mathcal{H}^t ]
\end{align*}

By similar argument as Theorem~\ref{thm:ff-scw}, we obtain \begin{equation*}
    \Exp [\SW (\mathcal A) ] \geq \frac{\OPT}{2\phi+2}
\end{equation*}
which concludes the proof.
\end{proof}

\begin{theorem}\label{thm:ff-sw-low-Fvalue}
Let $\A$ be the output of \FFCTRSDstar on a binary instance $I$ with $\phi<1$, where all agents report their friendships truthfully. Then ${\mathbb{E}}(\SW(\A))\ge \frac{\phi}{4\phi+4}\OPT(I)$.
\end{theorem}

\begin{proof}
The proof of the theorem is similar to the Theorem~\ref{thm:ff-scw}. Consider the same terminologies which were defined in Theorem~\ref{thm:ff-scw}. First we analyse the gain in social welfare at $t-$th iteration when $F^t \neq \emptyset $ and $\mathcal E^t \neq \phi $. We notice that the adjacent plots are available at $t-$th iteration. If randomly selected pair $(i,j)\in Z_1^t$ then both can obtain their maximum possible utility $1 + \phi$, however, it becomes little tricky when $(i,j)\notin Z_1^t$. If $(i,j)\in Z_2^t$, then increment in social welfare at $t-$th iteration is at least $\min \{2,2\phi +1\}$--either $(i,j)$ can grab the assigned plots in $\OPT^t$ or they both grab their respective high valued plots which are non-adjacent. Similarly, when $(i,j)\in Z^t_3\cup Z^t_4$, $(i,j)$ grabs two adjacent plots or two non-adjacent plots where at least one of them is getting high valued plots. Which implies;
\begin{equation*}
\begin{split}
      \Exp [\SW^{t+1} | \mathcal{H}^t ] \geq \SW^{t} + \frac{z_1^t(2\phi + 2)}{f^t} + \frac{z_2^t(\min\{2,2\phi +1\})}{f^t} \\
      + \frac{(z_3^t + z_4^t)(\min \{1,2\phi\})}{f^t}  
\end{split}
\end{equation*}

The decrease in the optimal welfare after $t-$th iteration should be upper bounded by a similar quantity as Theorem~\ref{thm:ff-scw}. Therefore for $\phi< 1$ we can write:

\begin{align*}
    &\Exp [\SW(\OPT^{t+1}) - \SW(\OPT^{t}) | \mathcal{H}^t ]\\
    &\leq \frac{z_1^t+z_2^t}{f^t}(6\phi + 4) + \frac{z_3^t}{f^t}(6\phi +2) + \frac{z_4^t}{f^t}(4\phi + 4)\\
    &\leq \frac{4\Big(1+\frac{1}{\phi} \Big)(2\phi + 2)z_1^t}{f^t} + \frac{4\Big(1+\frac{1}{\phi} \Big)(\min\{2,2\phi +1\})z_2^t}{f^t}\\
    &\qquad \qquad+ \frac{4\Big(1+\frac{1}{\phi} \Big)(\min\{2,2\phi +1\})( z_3^t + z_4^t)}{f^t}\\
    &\leq 4\Big(1+\frac{1}{\phi} \Big) \cdot \Exp [\SW^{t+1} - \SW^{t}  | \mathcal{H}^t ]
\end{align*}
 Now, by the similar analysis as Theorem~\ref{thm:on-ca-scw}, we obtain the desired result.

\end{proof}

\begin{theorem}
Let $\A$ be the output of \OnCARSDstar on a binary instance $I$ with $\phi<1$. Then ${\mathbb{E}}(\SW(\A))\ge \frac{\phi}{4\phi+4}\OPT(I)$.
\end{theorem}
\begin{proof}
The proof of the theorem is similar to the Theorem~\ref{thm:ff-scw} and Theorem~\ref{thm:ff-sw-low-Fvalue}.
\end{proof}

\subsection{Mixed Integer Program}

In this section, we showcase a mixed integer program (MIP) formulation for the \HAX problem.

\begin{align}
max: \sum_{i \in N} U_i(\A) \\
\sum_{i \in N} a_{i,v} &\leq 1,
        & \forall v \in \V                                  \label{eqn:one-per-buyer} \\
\sum_{v \in \V} a_{i,v} &\leq 1,
        & \forall i \in N                                   \label{eqn:one-per-cottage} \\
u_i = \sum_{v \in V} 
        [ u_i(v) &a_{i,v} + \phi \sum_{f \in F_i} \sum_{v' \in \N_v} a_{f,v'} ],
		& \forall i \in N								    \label{eqn:utility-fn} \\
a_{i,v} &\in \{0,1\}, 
        \forall i \in N, v \in \V                         \label{eqn:integrality} 
\end{align}

Equation~\ref{eqn:utility-fn} is the core of the MIP formulation, and it encodes the utility function by hard-coding the friends and nearby functions for each buyer and cottage.  Equations~\ref{eqn:one-per-buyer}~and~\ref{eqn:one-per-cottage} require that each buyer be matched to at most one cottage, and each cottage be matched to at most one buyer.  And Equation~\ref{eqn:integrality} encodes the binary variable constraints for $a_{b,c}$.

\end{document}